\documentclass{article}
\usepackage{paralist}
\usepackage{cite}
\usepackage{pgf,tikz}
\usetikzlibrary{arrows}
\usepackage{float}
\usepackage{algorithmic}
\usepackage{amssymb}
\usepackage{amsthm,amsmath}
\usepackage{tabularx}
\usepackage{hyperref}

\newcommand{\tw}{\ensuremath{\mathrm{t\!w}}}

\newtheorem{theorem}{Theorem}
\newtheorem{lemma}{Lemma}
\newtheorem{corollary}{Corollary}
\newtheorem{observation}{Observation}

\begin{document}

\title{Planar Induced Subgraphs of Sparse Graphs}

\author{Glencora Borradaile \and David Eppstein \and Pingan Zhu}


\maketitle
\begin{abstract}
We show that every graph has an induced pseudoforest of at least $n-m/4.5$ vertices, an induced partial 2-tree of at least $n-m/5$ vertices, and an induced planar subgraph of at least $n-m/5.2174$ vertices.  These results are constructive, implying linear-time algorithms to find the respective induced subgraphs.
We also show that the size of the largest $K_h$-minor-free graph in a given graph can sometimes be at most $n-m/6+o(m)$.
\end{abstract}

\pagestyle{plain} 

\section{Introduction}

\emph{Planarization}, a standard step in drawing non-planar graphs, involves replacing edge crossings with new vertices to form a planar graph with paths that represent the original graph's edges.  \emph{Incremental planarization}, does this by finding a large planar subgraph of the given graph, and then adding the remaining features of the input graph one at a time~\cite{BatEadTam-98}. Thus, it is of interest to study the algorithmic problem of finding planar subgraphs that are as large as possible in a given graph. Unfortunately, this problem is NP-hard and, more strongly, MAX-SNP-hard~\cite{CalFerFin-Algs-98}. A trivial algorithm, finding an arbitrary spanning tree, achieves an approximation ratio of~$\frac13$, and by instead searching for a partial 2-tree this ratio can be improved to~$\frac25$~\cite{CalFerFin-Algs-98}. The equivalent complementary problem, deleting a minimum number of edges to make the remaining subgraph planar, is fixed-parameter tractable and linear time for any fixed value of the parameter~\cite{KawRee-STOC-07}.

In this paper we study a standard variant of this problem: finding a large planar \emph{induced} subgraph of a given graph. In the context of the planarization problem, one possible application of finding this type of planar subgraph would be to apply incremental planarization in a drawing style where edges are represented as straight-line segments. A planar induced subgraph can always be drawn without crossings in this style, by F\'ary's theorem, after which the partial drawing could be used to guide the placement of the remaining vertices.
As with the previous problem, the induced planar subgraph problem is NP-hard, but again there is a linear-time fixed-parameter tractable algorithm for the equivalent problem of finding the smallest number of vertices to delete so that the remaining induced subgraph is planar~\cite{Kaw-FOCS-09}. 

Because of the difficulty of finding an exact solution to this problem, we instead seek worst-case guarantees: what is the largest size of a planar induced subgraph that we can guarantee to find within a graph of a given size? In this we are inspired by a paper of Alon, Mubayi, and Thomas~\cite{alon2001}, who showed that every \emph{triangle-free} input graph with $n$ vertices and $m$ edges contains an induced forest with at least $n-m/4$ vertices. This is tight, as shown by an input graph in the form of $n/4$ disjoint copies of a $4$-cycle. Induced forests are a special case of induced planar subgraphs, and so this result guarantees the existence of an induced planar subgraph of $n-m/4$ vertices. As we show, an analogous improvement to the one in the approximation ratio for planar subgraphs can be obtained by seeking instead an induced partial 2-tree. Rossmanith~\cite{BKMM13} has posed the question: does every graph have an induced planar subgraph of size $n-m/6$?  We shrink the gap on the worst-case bounds for the size of a planar induced subgraph by showing that every graph (not necessarily triangle-free) has an induced planar subgraph of $n-m/5.2174$ vertices and that there exist graphs for which the largest induced planar subgraph is not much larger than $n-m/6$ vertices.

\subsection{New results}

We prove the following results:

\begin{theorem} \label{thm:pseudoforest}
Every graph with $n$ vertices and $m$ edges has an induced pseudoforest with at least $n-\frac{2m}{9}$ vertices.
\end{theorem}

\begin{theorem} \label{thm:sp}
Every graph with $n$ vertices and $m$ edges has an induced subgraph with treewidth at most 2 and with at least $n-\frac{m}{5}$ vertices.
\end{theorem}

\begin{theorem} \label{thm:planar}
Every graph with $n$ vertices and $m$ edges has an induced planar subgraph with treewidth at most 3 and with at least $n-\frac{23m}{120}$ vertices.
\end{theorem}

\noindent These three theorems can be implemented as algorithms which take linear time to find the induced subgraphs described by the theorems.

\begin{theorem}\label{thm: lower_bound}
For every integer $h$, there is a family of graphs such that for any graph in this family with $n$ vertices and $m$ edges, the largest $K_h$-minor-free induced subgraph has at most $n-\frac{m}{6}+O(\frac{m}{\log m})$ vertices.
\end{theorem}

The bounds of \autoref{thm:pseudoforest} and \autoref{thm:sp} are tight, even for larger classes of induced subgraphs. In particular, there exist graphs for which the largest induced outerplanar subgraph has size at most $n-m/4.5$, so the bound of \autoref{thm:pseudoforest} is tight for any class of graphs between the pseudoforests and the outerplanar graphs. There also exist graphs for which the largest induced $K_4$-free induced subgraph has $n-m/5$ vertices, so \autoref{thm:sp} is tight for every family of graphs between the treewidth~2 graphs and the graphs with no $4$-clique.

\subsection{Related work}
The worst-case size of the largest induced planar subgraph has been studied previously by Edwards and Farr~\cite{EdwFar-GD-01}, who  proved a tight bound of $3n/(d+1)$ on its size as a function of the maximum degree~$d$ of the given graph. In contrast, by depending on the total number of edges rather than the maximum degree, our algorithms are sensitive to graphs with heterogeneous vertex degrees, and can construct larger induced subgraphs when the number of high-degree vertices is relatively small. Additionally, the algorithm given by Edwards and Farr is slower than ours, taking $O(mn)$ time. In a follow-up paper, Morgan and Farr~\cite{MorFar-JGAA-07} gave additional bounds on induced outerplanar subgraphs, and provided experimental results on the performance of their algorithms. A second paper by Edwards and Farr~\cite{EdwFar-DM-08}, like our \autoref{thm:sp}, gives bounds on the size of the largest induced partial 2-tree in terms of $n$ and $m$, which is $\frac{3n}{2m/n+1}$ for $m\ge 2n$. However, their bounds are asymptotically worse than \autoref{thm:sp} when $2n<m<2.5n$ and require an additional assumption of connectivity for smaller values of $m$. A third paper by the same authors~\cite{EdwFar-EJC-12} gives improved bounds that are more difficult to state as a formula.

For some other graph classes than the ones we study, it is possible to prove trivial bounds on the size of the largest induced subgraph in the class, of a similar form to the bounds of Alon et al. and of our theorems.
By repeatedly finding and removing a vertex of degree $\ge 1$, one can obtain an independent set of at least $n-m$ vertices, and the example of a perfect matching shows this to be tight. By repeatedly finding and removing a vertex of degree $\ge 2$, one can obtain a matching of at least $n-m/2$ vertices, and the example of the disjoint union of $n/3$ two-edge paths shows this to be tight. And by repeatedly finding and removing either a vertex of degree $\ge 3$ or a vertex that is part of a 2-regular cycle, one can obtain a linear forest (forest with maximum degree~2) of $n-m/3$ vertices; the example of the disjoint union of $n/3$ triangles shows this to be tight even for more general forests.

\autoref{tbl:known} provides a comparison of these new results with previous known results on induced planar subgraphs of various types and with the trivial bounds for independent sets, matchings, and linear forests.

\begin{table}[ph!]
\centering
\renewcommand{\arraystretch}{1.3}
\newcommand{\tbf}{\vrule height 15pt depth10pt width 0pt\displaystyle} 
\newcolumntype{A}{>{\centering\arraybackslash}X}
\begin{tabularx}{\textwidth}{|A|A|A|A|}
\hline
Size of induced subgraph & Constraints on $G$ & Type of subgraph & Reference\\\hline\hline
$\tbf n-m$& & independent set & trivial\\\hline
$\tbf n-\frac{m}{2}$& & matching & trivial\\\hline
$\tbf n-\frac{m}{3}$& & linear forest & trivial\\\hline
$\tbf n-\frac{m}{4}$& triangle-free & forest & \cite{alon2001}\\\hline
$\tbf 2\log \log n +O(\log \log \log n)$ & connected, $m = O(n)$ & tree & \cite{ESS86} \\\hline
$\tbf \frac{\log n}{4 \log r}$ & connected, $K_r$-free & tree & \cite{FLS09} \\\hline
$\tbf \sqrt{n}$ & connected, triangle-free & tree & \cite{FLS09} \\\hline
$\tbf \frac{\log n}{12\log \log n}$ & planar, 3-connected & path & \cite{GLM13} \\\hline
$\tbf n-\frac{m+c}{4}$& max degree $\le 3$ $c=\#$ connected components & forest & \cite{alon2001}\\\hline
$\tbf \frac{n}{\lceil (\Delta+1)/3 \rceil}$& max degree $\Delta$ & max degree 2 & \cite{HalLau-JGAA-97} \\\hline
$\tbf \frac{3n}{\Delta+5/3}$& max degree $\Delta$ & outerplanar & \cite{MorFar-JGAA-07} \\\hline
$\tbf \frac{3n}{\Delta+1}$& max degree $\Delta$ & planar & \cite{EdwFar-GD-01} \\\hline
$\tbf \frac{3n}{2m/n+1}$& $m\ge 2n$ or connected and $m\ge n$ & partial 2-tree & \cite{EdwFar-DM-08} \\ \hline
$\tbf \frac{5n}{6}$& claw-free subcubic& planar partial 4-tree & \cite{CheMcDSuz-WG-11} \\ \hline
$\tbf n-\frac{m}{4.5}$ & & pseudoforest & \autoref{thm:pseudoforest}\\\hline
$\tbf n-\frac{m}{5}$& & partial 2-tree & \autoref{thm:sp}\\\hline
$\tbf n-\frac{m}{5.2174}$& & planar partial 3-tree & \autoref{thm:planar}\\\hline
$\tbf \le n-\frac{m}{6}+o(m)$& & any minor-closed property & \autoref{thm: lower_bound} \\\hline
\end{tabularx}
\smallskip
\caption{Comparison of new and known results on induced subgraphs}
\label{tbl:known}
\end{table}

\section{Preliminaries}

For a graph $G$, we define $n(G)$ to be the number of vertices and
$m(G)$ to be the number of edges in~$G$. We drop the argument and write $n$ and $m$ when the choice of~$G$ is clear from context.

A subset $S$ of the vertices of $G$ corresponds to an \emph{induced subgraph} $G[S]$, a graph having $S$ as its vertices and having as edges every edge in $G$ that has both endpoints in $S$. Equivalently, $G[S]$ may be constructed from $G$ by deleting every vertex that is not in $S$ and every edge that has at least one endpoint outside $S$.

A \emph{pseudoforest} is an undirected graph in which every connected component has at most one cycle. Equivalently, the pseudoforests can be formed from forests (acyclic undirected graphs) by adding at most one edge per connected component. A \emph{$k$-tree} is an undirected graph that can be constructed from a $K_k$ graph by repeatedly picking a $K_k$ subgraph and attaching its $k$ vertices to a new vertex.  A \emph{partial $k$-tree} is a subgraph of a $k$-tree and is said to have {\em treewidth} at most $k$; the treewidth of a graph $G$ is denoted $\tw(G)$. Every pseudoforest is a partial 2-tree. A graph is a partial 2-tree if and only if every biconnected
component is a series parallel graph. The operations of adding a vertex with two adjacent neighbors and of taking subgraphs preserve planarity, so every partial 2-tree and every pseudoforest is a planar graph.

When constructing induced subgraphs of size $n-m/k$ for $k \in \mathbb{R}^+$, we will make the simplifying assumption that our graph $G$ has maximum degree at most $\lceil k-1 \rceil$.

\begin{observation}\label{obs:deg}
If every graph of maximum degree at most $\lceil k-1 \rceil$ contains an induced subgraph
with property $\cal P$ and at least $n-m/k$ vertices, then the same is true for every graph.
\end{observation}

\begin{proof}
We use induction on~$n$.
Let $G$ contain a vertex $v$ of degree $\ge k$, and let $G'$ be formed from $G$ by removing $v$. By the induction hypothesis,
$G'$ has an induced subgraph $H$ with the desired property $\cal P$ and at least $n(G')
- \frac{m(G')}{k}$ vertices.  Then $H$ is an induced $\cal P$-subgraph
of $G$ with size at least \[\hspace{4em}n(G)-1 - \frac{m(G')}{k} \ge n(G)-1 -
\frac{m(G)-k}{k} = n(G) - \frac{m(G)}{k}\hspace{4em} \]
\end{proof}

\section{Large induced pseudoforests}

In this section, we prove \autoref{thm:pseudoforest} by showing that we can delete at most $\frac{m}{4.5}$ vertices from a graph $G$ with $m$ edges to leave a pseudoforest; by \autoref{obs:deg}, we assume $G$ has degree at most 4.  

We repeatedly perform the first applicable reduction in the following list of cases, until no edges are left. As we do, we construct a set $S$ of vertices that will induce our desired subgraph.  Initially $S$ is empty and when no edges are left we add all remaining vertices to~$S$. The steps of the reduction essentially identify ``dangling trees'' and contract these.  We perform a series of vertex deletions (which identify vertices that will not belong to the final induced subgraph) and edge contractions.  If in doing so we create a component that consists of a single cycle (in fact a triangle; see case $\Delta$-a and {\em Vertex of degree 4} subcases (c), (i), and (ii)), we ``keep'' this component by adding the vertices of the cycle to $S$. This triangle is the minor of an induced cycle in our final induced subgraph which will only be incident to the ``dangling trees'' which had been contracted into the cycle.  This guarantees that the final induced subgraph has at most one cycle per component.  To bound the size of the output $S$ in terms of the number of edges, we use an amortized analysis, incurring a charge of $-4.5$ for every deleted vertex and a charge of $+1$ for every removed or contracted edge; we show that the net charge for every processing step is non-negative. Our cases are:

\newcommand{\ourcase}[1]{\smallskip\noindent{\bf #1.}~}

\ourcase{Leaf vertex}
If there is a vertex $a$ of degree 1, we add $a$ to $S$ and contract the edge incident to $a$.  This incurs a charge of $+1$.

\ourcase{Vertex of degree 2 not in a triangle}
If there is a vertex $a$ of degree 2
that is not part of a triangle, we add $a$ to $S$ and
contract an edge incident to $a$.  This incurs a charge of $+1$.

\ourcase{Vertex of degree 2 in a triangle}
If there is a vertex $a$
of degree 2 in a triangle $abc$ then we consider the four sub-cases
illustrated below:
\vspace{-.5cm}
\begin{enumerate}[($\Delta$--a)]
\item \begin{minipage}[t]{0.8\linewidth}
If the triangle is isolated: add $a$, $b$ and $c$ to $S$, and remove the edges of the triangle from the graph.  This incurs a charge of $+3$.
\end{minipage}
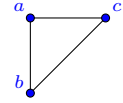
\begin{minipage}{0.1\linewidth}
 \raisebox{-6em} {\definecolor{xdxdff}{rgb}{0.49019607843137253,0.49019607843137253,1.0}
\definecolor{uuuuuu}{rgb}{0.26666666666666666,0.26666666666666666,0.26666666666666666}
\definecolor{qqqqff}{rgb}{0.0,0.0,1.0}
\begin{tikzpicture}[line cap=round,line join=round,>=triangle 45,x=1.0cm,y=1.0cm]
\draw (0.0,1.0)-- (0.0,-0.0);
\draw (1.0,1.0)-- (0.0,1.0);
\draw (1.0,1.0)-- (0.0,0.0);
\begin{scriptsize}
\draw [fill=qqqqff] (0.0,-0.0) circle (1.5pt);
\draw[color=qqqqff] (-0.15,0.15) node {$b$};
\draw [fill=qqqqff] (0.0,1.0) circle (1.5pt);
\draw[color=qqqqff] (-0.15,1.15) node {$a$};
\draw [fill=qqqqff] (1.0,1.0) circle (1.5pt);
\draw[color=qqqqff] (1.15,1.15) node {$c$};
\end{scriptsize}
\end{tikzpicture}}
\end{minipage}
\vspace{-0.8cm}
\item \begin{minipage}[t]{0.63\linewidth}
  If $b$ has degree 2 and $c$ has degree 3, then $c$ is adjacent to a
  vertex $d$ of degree at least 3 (otherwise $d$ would be a degree 2 vertex not belonging to a triangle). Delete $d$, isolating
  triangle $abc$, and then apply case ($\Delta$--a). This incurs a total charge
  of at least $+1.5$.
\end{minipage}
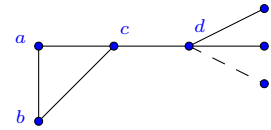
\begin{minipage}[c]{0.3\linewidth}
\raisebox{-9em}{\definecolor{qqqqff}{rgb}{0.0,0.0,1.0}
\definecolor{ttttff}{rgb}{0.2,0.2,1.0}
\begin{tikzpicture}[line cap=round,line join=round,>=triangle 45,x=1.0cm,y=1.0cm]
\draw (0.0,-1.0)-- (0.0,0.0);
\draw (0.0,0.0)-- (1.0,0.0);
\draw (1.0,0.0)-- (0.0,-1.0);
\draw (1.0,0.0)-- (2.0,0.0);
\draw (2.0,0.0)-- (3.0,0.5);
\draw (2.0,0.0)-- (3.0,0.0);
\draw [dash pattern=on 5pt off 5pt] (2.0,0.0)-- (3.0,-0.5);
\begin{scriptsize}
\draw [fill=qqqqff] (0.0,0.0) circle (1.5pt);
\draw[color=qqqqff] (-0.24,0.09999999999999888) node {$a$};
\draw [fill=qqqqff] (0.0,-1.0) circle (1.5pt);
\draw[color=qqqqff] (-0.24,-0.9400000000000014) node {$b$};
\draw [fill=qqqqff] (1.0,0.0) circle (1.5pt);
\draw[color=qqqqff] (1.1400000000000001,0.2099999999999989) node {$c$};
\draw [fill=qqqqff] (2.0,0.0) circle (1.5pt);
\draw[color=qqqqff] (2.14,0.2799999999999989) node {$d$};
\draw [fill=qqqqff] (3.0,0.0) circle (1.5pt);
\draw [fill=qqqqff] (3.0,0.5) circle (1.5pt);
\draw [fill=qqqqff] (3.0,-0.5) circle (1.5pt);
\end{scriptsize}
\end{tikzpicture}}
\end{minipage}
\item \begin{minipage}[t]{0.75\linewidth}
If $b$ has degree 3 and $c$ has degree at least 3, then delete $c$, add $a$ and $b$ to $S$ and contract the edges incident to $b$.  This incurs a charge of at least $+0.5$.
\end{minipage}
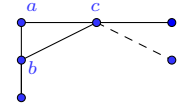
\begin{minipage}[c]{0.25\linewidth}
 \raisebox{-5.5em}{ \definecolor{qqqqff}{rgb}{0.0,0.0,1.0}
\definecolor{ttttff}{rgb}{0.2,0.2,1.0}
\begin{tikzpicture}[line cap=round,line join=round,>=triangle 45,x=1.0cm,y=1.0cm]
\draw (0.0,-0.0)-- (0.0,-1.0);
\draw (0.0,-0.0)-- (1.0,0.0);
\draw (1.0,0.0)-- (0.0,-0.5);
\draw (1.0,0.0)-- (2.0,0.0);
\draw (0.0,-0.5)-- (0.0,-1.0);
\draw [dashed] (1.0,0.0) -- ( 2.0,-0.5 );
\begin{scriptsize}
\draw [fill=ttttff] (0.0,-0.0) circle (1.5pt);
\draw[color=ttttff] (0.13999999999999996,0.2) node {$a$};
\draw [fill=ttttff] (0.0,-0.5) circle (1.5pt);
\draw[color=ttttff] (0.15,-0.6) node {$b$};
\draw [fill=ttttff] (1.0,0.0) circle (1.5pt);
\draw[color=ttttff] (0.9799999999999996,0.2) node {$c$};
\draw [fill=qqqqff] (2.0,0.0) circle (1.5pt);
\draw [fill=ttttff] (0.0,-1.0) circle (1.5pt);
\draw [fill=ttttff] (2.0,-0.5) circle (1.5pt);
\end{scriptsize}
\end{tikzpicture}}
\end{minipage}
\item \begin{minipage}[t]{0.75\linewidth}
   If $b$ has degree 4, then delete $b$, add $a$ to $S$ and contract $ac$.  This incurs a charge of $+0.5$. 
\end{minipage}
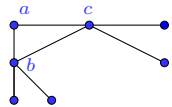
\begin{minipage}[c]{0.25\linewidth}
  \raisebox{-5em}{\definecolor{qqqqff}{rgb}{0.0,0.0,1.0}
\definecolor{ttttff}{rgb}{0.2,0.2,1.0}
\begin{tikzpicture}[line cap=round,line join=round,>=triangle 45,x=1.0cm,y=1.0cm]
\draw (0.0,-0.0)-- (0.0,-1.0);
\draw (0.0,-0.0)-- (1.0,0.0);
\draw (1.0,0.0)-- (0.0,-0.5);
\draw (1.0,0.0)-- (2.0,0.0);
\draw (0.0,-0.5)-- (0.0,-1.0);
\draw (1.0,0.0) -- ( 2.0,-0.5 );
\draw (0.0, -0.5) -- (0.5, -1.0);
\begin{scriptsize}
\draw [fill=ttttff] (0.0,-0.0) circle (1.5pt);
\draw[color=ttttff] (0.13999999999999996,0.2) node {$a$};
\draw [fill=ttttff] (0.0,-0.5) circle (1.5pt);
\draw[color=ttttff] (0.225,-0.53) node {$b$};
\draw [fill=ttttff] (1.0,0.0) circle (1.5pt);
\draw[color=ttttff] (0.9799999999999996,0.2) node {$c$};
\draw [fill=qqqqff] (2.0,0.0) circle (1.5pt);
\draw [fill=ttttff] (0.0,-1.0) circle (1.5pt);
\draw [fill=ttttff] (2.0,-0.5) circle (1.5pt);
\draw [fill=ttttff] (0.5, -1.0) circle (1.5pt);
\end{scriptsize}
\end{tikzpicture}}
\end{minipage}
\end{enumerate}

\ourcase{Vertex of degree 3 adjacent to a vertex of degree 4}  If there is a vertex $a$ of degree 3 adjacent to a vertex $b$ of degree 4, then we delete $b$.  Deleting $b$ incurs a net charge of $-0.5$, but reduces the degree of $a$ to 2.  Handling $a$ as above incurs a charge of at least $+0.5$, for a net charge of at least 0.

\ourcase{Vertex of degree 3} If this is the first applicable case, then the graph must be 3-regular. Deleting any vertex $a$ creates three vertices of degree~2 while incurring a charge of $-1.5$. Processing the three resulting degree-2 vertices as above incurs a charge of at least $+0.5$ per degree-2 vertex, for a net charge of at least 0.

\ourcase{Vertex of degree 4} If this is the first applicable case, then the graph must be 4-regular. We consider the following cases for the subgraph $N_a$ induced by the neighbors $b,c,d,e$ of a vertex~$a$.  
\begin{enumerate}[(a)]
\item If $N_a$  has two non-adjacent pairs $(b,c)$ and $(d,e)$, then $a$, $b$, and $c$ do not form a triangle, so we can delete $d$ and $e$ and contract $ab$. This removes nine edges from the graph and deletes two vertices, for a total charge of $0$.

\item If $N_a$ is a star graph with center $b$, then consider the
  neighbors of $c$: $a, b, f, g$. Neither $a$ nor $b$ can be adjacent
  to $f$ or $g$, because otherwise they would have too many neighbors.  Thus we can process vertex $c$ as in case (a) instead.

\item $N_a$ contains a triangle if and only if neither case (a) nor (b) applies: if $N_a$ contains at most 2 edges, then case (a) applies; if $N_a$ contains three edges but no triangle, then $N_a$ is a star (and case (b) applies) of $N_a$ is a path of length 3 (and case (a) applies); if $N_a$ contains four or more edges, then it must contain a triangle.

Without loss of generality the triangle is formed by vertices
  $bcd$, so $a$ is a vertex of a
  tetrahedron ($K_4$) induced by vertices $a,b,c,d$.  We may assume more strongly that every vertex in the graph belongs to a
  tetrahedron, for if not we may apply cases (a) or (b).  We form four sub-subcases:
  \begin{enumerate}[(i)]
  \item If any connected component is a complete graph  $K_5$, then deleting  two vertices leaves an isolated triangle and incurs a total charge of $+1$.
  \item If two tetrahedra, $a,b,c,d$ and $b,c,d,e$
    share triangle $b,c,d$ without forming a $K_5$, then $a$ and $e$ are non-adjacent.
    Deleting $a$ and $e$ removes eight edges but
    leaves the isolated triangle $b,c,d$ for a net charge of~$+2$ (see case~$\Delta$--a). We illustrate this case:
    \begin{center}
      \definecolor{qqqqff}{rgb}{0.0,0.0,1.0}
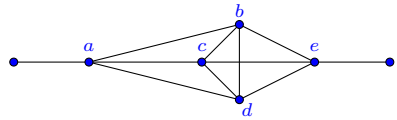
\begin{tikzpicture}[line cap=round,line join=round,>=triangle 45,x=1.0cm,y=1.0cm]
\draw (0.0,0.5)-- (-0.5,-0.0);
\draw (0.0,0.5)-- (1.0,0.0);
\draw (0.0,-0.5)-- (1.0,0.0);
\draw (0.0,-0.5)-- (-0.5,-0.0);
\draw (1.0,0.0)-- (2.0,0.0);
\draw (-0.5,-0.0)-- (-2.0,0.0);
\draw (-2.0,0.0)-- (0.0,0.5);
\draw (-2.0,0.0)-- (0.0,-0.5);
\draw (-2.0,0.0)-- (-3.0,0.0);
\draw (0.0,0.5)-- (0.0,-0.5);
\draw (-0.5,-0.0)-- (1.0,0.0);
\begin{scriptsize}
\draw [fill=qqqqff] (-0.5,-0.0) circle (1.5pt);
\draw [fill=qqqqff] (0.0,-0.5) circle (1.5pt);
\draw [fill=qqqqff] (-2.0,0.0) circle (1.5pt);
\draw[color=qqqqff] (-2,0.2) node {$a$};
\draw [fill=qqqqff] (1.0,0.0) circle (1.5pt);
\draw[color=qqqqff] (1.,0.2) node {$e$};
\draw [fill=qqqqff] (0.0,0.5) circle (1.5pt);
\draw [fill=qqqqff] (2.0,0.0) circle (1.5pt);
\draw [fill=qqqqff] (-3.0,0.0) circle (1.5pt);
\draw[color=qqqqff] (-0.5,0.2) node {$c$};
\draw[color=qqqqff] (0,0.68) node {$b$};
\draw[color=qqqqff] (0.1,-0.63) node {$d$};
\end{scriptsize}
\end{tikzpicture}
    \end{center}
      \item In the remaining cases all tetrahedra must be vertex-disjoint, for if two tetrahedra share only an edge or a vertex, there would be a vertex with degree $>4$. If  two tetrahedra $abcd$ and $efgh$ are connected to each other by at least two edges ($be$ and $dg$), then we delete the two non-adjacent vertices
        $d$ and $e$ as illustrated here:

\hspace{2em}
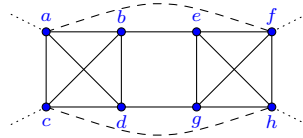
\begin{minipage}{0.35\linewidth}
\definecolor{qqqqff}{rgb}{0.0,0.0,1.0}
\begin{tikzpicture}[line cap=round,line join=round,>=triangle 45,x=1.0cm,y=1.0cm]
\draw (-1.5,-0.5)-- (-1.5,0.5);
\draw (-1.5,0.5)-- (-0.5,0.5);
\draw (-0.5,0.5)-- (0.5,0.5);
\draw (0.5,0.5)-- (1.5,0.5);
\draw (1.5,0.5)-- (1.5,-0.5);
\draw (1.5,-0.5)-- (0.5,-0.5);
\draw (0.5,-0.5)-- (-0.5,-0.5);
\draw (-0.5,-0.5)-- (-1.5,-0.5);
\draw (-0.5,0.5)-- (-0.5,-0.5);
\draw (0.5,0.5)-- (0.5,-0.5);
\draw (-0.5,0.5)-- (-1.5,-0.5);
\draw (-1.5,0.5)-- (-0.5,-0.5);
\draw (1.5,0.5)-- (0.5,-0.5);
\draw (0.5,0.5)-- (1.5,-0.5);
\draw[dashed] (-1.5,0.5) .. controls (0,1) .. (1.5, 0.5);
\draw[dashed] (-1.5,-0.5) .. controls (0,-1) .. (1.5, -0.5);
\draw[dotted] (-1.5,0.5) -- (-2,0.75);
\draw[dotted] (-1.5,-0.5) -- (-2,-0.75);
\draw[dotted] (1.5,0.5) -- (2,0.75);
\draw[dotted] (1.5,-0.5) -- (2,-0.75);
\begin{scriptsize}
\draw [fill=qqqqff] (-0.5,0.5) circle (1.5pt);
\draw[color=qqqqff] (-0.5,0.7) node {$b$};
\draw [fill=qqqqff] (-0.5,-0.5) circle (1.5pt);
\draw[color=qqqqff] (-0.5,-0.7) node {$d$};
\draw [fill=qqqqff] (-1.5,0.5) circle (1.5pt);
\draw[color=qqqqff] (-1.5,0.7) node {$a$};
\draw [fill=qqqqff] (-1.5,-0.5) circle (1.5pt);
\draw[color=qqqqff] (-1.5,-0.7) node {$c$};
\draw [fill=qqqqff] (0.5,0.5) circle (1.5pt);
\draw[color=qqqqff] (0.5,0.7) node {$e$};
\draw [fill=qqqqff] (0.5,-0.5) circle (1.5pt);
\draw[color=qqqqff] (0.5,-0.7) node {$g$};
\draw [fill=qqqqff] (1.5,0.5) circle (1.5pt);
\draw[color=qqqqff] (1.5,0.7) node {$f$};
\draw [fill=qqqqff] (1.5,-0.5) circle (1.5pt);
\draw[color=qqqqff] (1.5,-0.7) node {$h$};
\end{scriptsize}
\end{tikzpicture}
\end{minipage}
\begin{minipage}{0.6\linewidth}
  \begin{quote} {\small The dashed edges are possible connections from
      vertices $a,c,f,h$.}
  \end{quote}
\end{minipage}

        This incurs a charge of $-1$ but leaves two
        non-adjacent degree two vertices ($b$ and $g$) each of which can be processed via case
        ($\Delta$--c), adding charge $+0.5$ per vertex for a total charge of at least 0.
  \item If every pair of tetrahedra are connected by at most one
    edge, then contracting every tetrahedron to a vertex reduces the input graph to a smaller
    4-regular simple graph that necessarily contains a cycle of three or more edges. In the uncontracted graph, this gives a cycle of six or more edges that alternates between edges within tetrahedra and edges outside the tetrahedra:
   \begin{center}
      \definecolor{qqqqff}{rgb}{0.0,0.0,1.0}
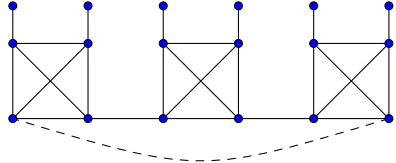
\begin{tikzpicture}[line cap=round,line join=round,>=triangle 45,x=1.0cm,y=1.0cm]
\draw (-2.0,0.0)-- (-2.0,1.0);
\draw (-2.0,1.0)-- (-1.0,1.0);
\draw (-1.0,1.0)-- (-1.0,0.0);
\draw (-1.0,0.0)-- (-2.0,0.0);
\draw (-2.0,0.0)-- (-1.0,1.0);
\draw (-2.0,1.0)-- (-1.0,0.0);
\draw (-1.0,0.0)-- (0.0,0.0);
\draw (0.0,1.0)-- (0.0,0.0);
\draw (0.0,1.0)-- (1.0,1.0);
\draw (1.0,1.0)-- (1.0,0.0);
\draw (1.0,0.0)-- (0.0,0.0);
\draw (0.0,0.0)-- (1.0,1.0);
\draw (0.0,1.0)-- (1.0,0.0);
\draw [dashed](3,0) .. controls (0.5,-0.75) .. (-2, 0);
\draw (-2,1) -- (-2,1.5);
\draw (-1,1) -- (-1,1.5);
\draw (0,1) -- (0,1.5);
\draw (1,1) -- (1,1.5);
\draw (1,0) -- (2.0,0);
\draw (2,0) -- (2.0,1);
\draw (2.0,1) -- (3.0,1);
\draw (3.0,1) -- (3.0,0);
\draw (3.0,0) -- (2.0,0);
\draw (2,1) -- (3.0,0);
\draw (3,1) -- (2.0,0);
\draw (2,1) -- (2.0,1.5);
\draw (3,1) -- (3.0,1.5);
\begin{scriptsize}
\draw [fill=qqqqff] (-1.0,0.0) circle (1.5pt);
\draw [fill=qqqqff] (-1.0,1.0) circle (1.5pt);
\draw [fill=qqqqff] (-2.0,1.0) circle (1.5pt);
\draw [fill=qqqqff] (0.0,1.0) circle (1.5pt);
\draw [fill=qqqqff] (1.0,1.0) circle (1.5pt);
\draw [fill=qqqqff] (0.0,0.0) circle (1.5pt);
\draw [fill=qqqqff] (1.0,0.0) circle (1.5pt);
\draw [fill=qqqqff] (-2.0,0) circle (1.5pt);
\draw [fill=qqqqff] (-2,1.5) circle (1.5pt);
\draw [fill=qqqqff] (-1,1.5) circle (1.5pt);
\draw [fill=qqqqff] (0,1.5) circle (1.5pt);
\draw [fill=qqqqff] (1,1.5) circle (1.5pt);
\draw [fill=qqqqff] (2.0,0) circle (1.5pt);
\draw [fill=qqqqff] (2.0,1) circle (1.5pt);
\draw [fill=qqqqff] (3.0,1) circle (1.5pt);
\draw [fill=qqqqff] (3.0,0) circle (1.5pt);
\draw [fill=qqqqff] (2.0,1.5) circle (1.5pt);
\draw [fill=qqqqff] (3.0,1.5) circle (1.5pt);

\end{scriptsize}
\end{tikzpicture}
    \end{center}
 In this case, we choose one of the tetrahedra of the cycle, and delete the two vertices of this tetrahedron that do not belong to the cycle. This removes seven edges from the graph for a net    charge of $-2$, but leaves two degree-2 vertices on a cycle of length at least~6. Each of these may be processed as a degree-2 vertex that does not belong to a triangle, giving a charge of $+1$ each and making the net charge be~0.

 \end{enumerate}
\end{enumerate}

This case analysis concludes the proof of \autoref{thm:pseudoforest}.
The proof also gives the outline for an efficient algorithm for finding an induced pseudoforest of size at least $n-m/4.5$: after removing any high-degree vertices, form a data structure that lists the configurations of the graph obeying each of the cases in the analysis. Because the remaining graph has bounded degree, selecting the first applicable case, performing the reduction steps of the case, and updating the list of configurations for each case can all be done in constant time per case, leading to a linear overall time bound.

\autoref{thm:pseudoforest} is tight: there exist arbitrarily large graphs in which the largest induced pseudoforest has exactly $n-m/4.5$ vertices. In particular, let $G$ be a graph formed by the disjoint union of $n/6$ copies of the complete bipartite graph $K_{3,3}$.
Then, to form a pseudoforest in $G$, we must delete at least two vertices from each copy of $K_{3,3}$, for deleting only one vertex leaves $K_{2,3}$ which is not a pseudoforest.
Each copy has nine edges, so the number of deleted vertices must be at least $m/4.5$. The same class of examples shows that even if we are searching for the broader class of induced outerplanar subgraphs, we may need to delete $m/4.5$ vertices.

\section{Large induced treewidth two graphs}

In this section, we prove \autoref{thm:sp} by showing that we can delete at most $\frac{m}{5}$ vertices from a graph $G$ with $m$ edges to leave a graph with treewidth at most~2. By \autoref{obs:deg}, we may assume without loss of generality $G$ has degree at most 4.  We prove the theorem algorithmically by arguing that the following procedure builds a vertex set $S$ of size at least $n-\frac{m}{5}$ such that $G[S]$ has treewidth at most~2.  The procedure modifies the graph by edge contractions but does not increase its degree over~4.
\newpage
\begin{tabbing}
  $S = \emptyset$\\
  make $G$ simple by removing self-loops and parallel edges \\
  while $G$ has more than 1 vertex:\\
  \qquad\= if there is a vertex $v$ of degree one or two: \\
  \> \qquad \= contract an edge incident to $v$ and add $v$ to $S$: \\
  \> \> make $G$ simple by removing self-loops and parallel edges \\
  \> else if $G$ contains a vertex of degree three:\\
  \> \> delete a vertex of the largest degree adjacent to a degree-three vertex  \\
  \> else: \\
  \> \> delete a vertex of maximum degree \\
  add the isolated vertices to $S$
\end{tabbing}

\begin{lemma}\label{lem: algorithm induces SP}
The induced subgraph $G[S]$ produced by the algorithm above has treewidth at most~2.
\end{lemma}

\begin{proof}
  Mark the edges that are contracted and the edges that are removed as self-loops and parallel edges.  Then all the edges of $G[S]$ are colored.  (Edges not in $G[S]$ may also be colored.)  

  Recall that a graph is series parallel if can be reduced to a set of isolated vertices by contracting edges incident to vertices of degree 2 and deleting parallel edges and loops.  Consider the order $e_1, e_2, \ldots, $ of the edges of $G[S]$ by the order in which they are contracted or removed from $G$.  Since $G[S]$ only consists of colored edges, for every $i$, $e_i$ must be a parallel edge, a self-loop or incident to a vertex of degree 2 in $G[\cup_{j \ge i} e_j]$ and so has the same property in $G[S]$ which is a subgraph.  Therefore, $G[S]$ is series parallel, and so, equivalently, has treewidth 2. 
\end{proof}





\begin{lemma}\label{lem: s in algorithm less than m/5}
$|S| \ge n-\frac{m}{5}$.
\end{lemma}
\begin{proof}
  We show, equivalently, that the procedure deletes at most
  $\frac{m}{5}$ vertices by amortized analysis.  For each vertex that
  we delete we incur a charge of $-5$.  For each edge that we contract
  (incident to a degree-1 or degree-2 vertex), remove (as a self-loop or
  parallel edge) or delete (by way of deleting an adjacent vertex) we
  incur a charge of $+1$.  Note that we distinguish between deleting
  and removing an edge for the purpose of this analysis.  We will show
  that the net charge is positive, thus showing that for every 5 edges
  of the graph, we delete at most one vertex.

 The first case of the algorithm, in which an edge is contracted, incurs only a positive charge. There are three remaining  cases in which a vertex is deleted, according to the degree of the deleted vertex
  and whether it is adjacent to a degree 3 vertex.  

  \ourcase{Deleting a degree-3 vertex from a 3-regular graph}
  Deleting such a vertex incurs a charge of $+3-5 = -2$ but creates
  three degree-2 vertices. At least one edge incident to each of these
  will be contracted (or removed) for a total charge of $+3$. Therefore, before another
  vertex is deleted, the net charge for deleting this vertex is at
  least $+1$.

  \ourcase{Deleting a degree-4 vertex adjacent to a degree-3 vertex} 
  Deleting such a vertex $v$ incurs a charge of $+4-5 = -1$ but
  creates at least one vertex $u$ of degree 2; an edge incident to $u$
  will be contracted before another vertex is deleted.  The net charge
  for deleting $v$ is at least 0.

  \ourcase{Deleting a degree-4 vertex from a 4-regular graph} Deleting such a vertex incurs a charge of $+4-5 = -1$. After this case happens, the remaining graph will have at least four degree-3 vertices, and will continue to have a mix of degree-3 and degree-4 vertices until either the graph becomes 4-regular again or all degree-4 vertices have been removed.  Until this happens, the operations that occur are the deletion of degree-4 vertices incident to degree-3 vertices (incurring non-negative charge, as described in the previous case) and the contraction of edges.  We show that there will be an operation that occurs a positive charge before the procedure is complete or before another negative charge can be incurred (if the graph returns to being 4-regular), balancing the negative charge from the deletion of the degree-4 vertex from the 4-regular graph.
There are 3 cases:
  \begin{enumerate}[(a)]
  \item If the graph becomes 4-regular, it
    can only be after removing a degree-4 vertex adjacent to four
    degree-3 vertices, followed by four contractions of the resulting
    degree-2 vertices. These contractions ($+4$ charge) and deletion
    ($-5$ for deleting the degree-4 vertex $+4$ for the deletion of
    the incident edges) give a charge of $+3$.
  \item If all the degree-4 vertices are deleted without making the graph 3-regular, then the last degree-4 vertex deleted must not have had any degree-4 neighbors. The deletion of the last degree-4 vertex causes its neighbors to have degree~2 and is followed by four contractions, for a charge of $+3$.
  \item If the graph becomes 3-regular, then the deletion of the first
    degree-3 vertex from this 3-regular state incurs a charge of $+1$
    (as argued in the first case).
  \end{enumerate}
  Thus, in all cases, the negative charge for the removal of a degree-4 vertex from a 4-regular graph is balanced by a positive charge for a subsequent step of the algorithm.  
\end{proof}

The bound $n-m/5$ is tight, for arbitrarily large graphs. In particular, the graphs formed from $n/5$ disjoint copies of the complete graph $K_5$ can have at most $n-m/5$ vertices in any induced subgraph of treewidth at most~2, for otherwise one of the copies of $K_5$ would have only one of its vertices removed in the subgraph, leaving a $K_4$ subgraph which does not have treewidth~2.

\section{Large induced planar graphs}

In this section, we prove \autoref{thm:planar}. In outline, the proof (and corresponding algorithm) is similar to that of \autoref{thm:sp}, but with more cases and a more complex system of charges inspired by the ``measure and conquer'' method for analyzing exponential-time backtracking algorithms~\cite{FomGraKra-ICALP-05}. To handle the complexity of the analysis, we will set up a linear program
whose variables are certain parameters of the analysis, the foremost of these being a number $\epsilon$ such that the algorithm finds a subgraph of at least $n-m/(5+\epsilon)$ vertices.
These parameters are used only in the analysis of the algorithm; the algorithm itself does not need to know them, but instead (as with our other algorithms) merely chooses the first available reduction case. 

Similarly to our previous charging schemes for analyzing our algorithms,
we will charge $-(5+\epsilon)$ for every removed vertex and $+1$ for every removed edge. However,  instead of showing that some steps of the algorithm with a negative charge are balanced by later steps with a positive charge, we require that each step must immediately have zero total charge. In order to achieve this, we augment our analysis by a system of  ``debts'' that will have to be paid later. At any step, we can place a debt on a vertex or on the whole graph. This causes a net positive charge for that step, in the amount of the new debt, but a later step that removes this vertex from the graph must incur a negative charge equal to the debt of the vertex in addition to its other charges.

We give a vertex of degree $d$ a credit limit of $c_d$ units; this defines the maximum debt that can be placed on it. We will set
\begin{equation}
1=c_2\ge c_3\ge c_4\ge 0\label{eq:1}
\end{equation}
and set $c_d=0$ for all $d>4$. We define $\delta_i$ (for $i\in\{2,3,4\}$) to equal $c_i-c_{i+1}$, and we will constrain these credit limits to ensure that
\begin{equation}
\delta_2\ge\delta_3\ge\delta_4.
\end{equation}
When a step causes the degree of a vertex $v$ to decrease from $d+1$ to $d$, we may place $\delta_d$ units of additional debt on that vertex, but when we delete or remove a vertex of degree $d$ we may have to pay $c_d$ units to clear its debt. In addition, we use one more parameter, $\tau$, giving the credit limit for the whole graph. The graph may incur a debt of $\tau$ units in a step in which the number of degree-three vertices becomes nonzero, but must clear this debt (subtracting $\tau$ from the charge for the step) at any step that removes all remaining degree-three vertices.

As in our previous algorithms we may assume without loss of generality that all vertices have degree at most five, by removing higher-degree vertices until no more such vertices remain. We then have the following cases:

\ourcase{Degree two} If there is a degree two vertex, we contract an edge and add the vertex to~$S$. As we have set $c_2 = 1$, the charge for clearing any debt on the vertex is canceled by the charge for the contracted edge, giving non-negative total charge.

\ourcase{Planar} If the remaining graph is either $K_4$ or a two-vertex multigraph with three edges (formed from $K_{3,3}$ by deleting a vertex and contracting degree-two paths), we add all its vertices to $S$. If there are $k$ vertices, the total debt is at most $kc_3+\tau$ and the total charge for the edges in the remaining graph is $3k/2$. The worst case happens for $k=2$, giving a total charge that is non-negative as long as:
\begin{equation}
2c_3+\tau\le 3\label{eq:2}
\end{equation}

\ourcase{3-regular} If the graph is 3-regular we delete a vertex. Because the previous (planar) case did not apply, it must have had six or more vertices prior to the deletion, so the deletion cannot eliminate all degree-three vertices. The deletion gives us a charge
of $-(5+\epsilon)-c_3$, the three deleted edges give a charge of $+3$, and we can place $\delta_2=1-c_3$ more units of debt on each of the three neighbors, giving us a net charge of $1-\epsilon-4c_3$. The result is non-negative as long as:
\begin{equation}
4c_3+\epsilon\le 1\label{eq:3}
\end{equation}

\ourcase{Degree 5} If there exists a degree-5 vertex we delete it, and increase the debt on its neighbors by $c_4=\delta_4$ each. This incurs a net charge of $-\epsilon+5c_4$. The result is non-negative as long as: 
\begin{equation}
\epsilon-5c_4\le 0\label{eq:4}
\end{equation}

\ourcase{Mixed degrees} If there exists a degree-4 vertex $v$ adjacent to a degree-3 vertex, we delete $v$. 
There are two cases:
\begin{enumerate}[(a)]
\item After the deletion, the graph still has degree-3 vertices remaining.
The charge for deleting the vertex and its four edges is $-1-\epsilon$, it may already carry a debt of $c_4$, and the amount of additional debt that can be placed on its neighbors is at least $\delta_2+3\delta_3$. The result is non-negative as long as:
\begin{equation}
\epsilon - 2c_3 + 4c_4\le 0\label{eq:5}
\end{equation}

\item The deletion eliminates all remaining degree-3 vertices. This can only happen if all neighbors of the deleted vertex had degree three. The charge for deleting the vertex and its four edges is $-1-\epsilon$, and it may already carry a debt of $c_4$. An additional charge of $4\delta_2$ may be placed on its four neighbors, but we must also pay a charge of $-\tau$ for the elimination of all degree-three vertices. Thus, the total charge is
$-1-\epsilon-c_4+4(1-c_3)-\tau=3-\epsilon-4c_3-c_4-\tau$. The result is non-negative as long as:
\begin{equation}
\epsilon+4c_3+c_4+\tau\le 3\label{eq:6}
\end{equation}

\end{enumerate}

\ourcase{4-regular} If the graph is 4-regular, we delete any vertex. The charge for deleting the vertex and its four edges is $-1-\epsilon$, it may carry a debt of $c_4$, and additionally we place $\delta_3$ units of debt at each neighbor and $\tau$ units of debt on the whole graph. Thus, the total charge is $-\epsilon-1+\tau+4\delta_3-c_4$. The result will be non-negative as long as:
\begin{equation}
 \epsilon-\tau-4c_3+5c_4\le -1\label{eq:7}
\end{equation}
\begin{table}[t]
{\small\begin{verbatim}epsilon;
2 c3 - c4 <= 1; /* delta_2 >= delta_3 */
c3 - 2 c4 >= 0; /* delta_3 >= delta_4 */
c4 >= 0;        /* delta_4 >= 0 */
tau >= 0;
2 c3 + tau <= 3;                    /* case: planar */
4 c3 + epsilon <= 1;                /* case: 3-regular */
epsilon - 5 c4 <= 0;                /* case: degree 5 */
epsilon - 2 c3 + 4 c4 <= 0;         /* case: mixed degrees (a) */
epsilon + 4 c3 + c4 + tau <= 3;     /* case: mixed degrees (b) */
epsilon - tau - 4 c3 + 5 c4 <= -1;  /* case: 4-regular */\end{verbatim}}
\caption{Formulation as a linear program of the problem of maximizing $\epsilon$ given the linear inequalities on $\epsilon$, $c_3$, $c_4$, and $\tau$ required by our case analysis, in the .lp format of the lp\_solve free linear programming solver.}
\label{tbl:lp}
\end{table}

We applied the free lp\_solve linear programming solver to the system of inequalities arising from these cases, given by Equations~(\ref{eq:1}) to~(\ref{eq:7}); \autoref{tbl:lp} provides the input given to this solver. We then used a method based on continued fractions to find the exact rational solution values corresponding to the decimal numbers provided by the solver. The optimal solution found in this way has the values $\epsilon=5/23$, $c_3=9/46$, $c_4=1/23$, and $\tau=15/23$. It is straightforward to plug these values back into the given constraints and verify that they provide a valid solution.

To prove that the resulting induced subgraph is planar and has treewidth at most 3, we observe that most steps of the algorithm modify the input only by deleting a vertex; the exceptions are case of a {\em planar} graph (a connected component of the remaining graph isomorphic to $K_4$ or to a two-vertex three-edge multigraph) and the case of a degree-two vertex (which we contract, adding the contracted-away vertex to the eventual solution). Thus, each connected component of the output graph can be formed from $K_4$ or from the three-edge multigraph by reversing the degree-two contraction steps. Each reversal of a contraction replaces an edge by a two-edge path, so each component of the output must be a subdivision of $K_4$ or of the three-edge multigraph. In both cases the result is planar and has treewidth at most three.

\section{No very large, minor-free induced subgraphs}

In this section we prove \autoref{thm: lower_bound}.  To prove
this theorem, we begin with the well-known result that $K_h$-minor
free graphs are sparse~\cite{kostochka1982minimum,
  kostochka1984lower,thomason1984extremal, thomason2001extremal}.
\begin{lemma}[Theorem 1.1~\cite{thomason2001extremal}]\label{lem: Kost1984}
Every simple $K_h$-minor-free graph with $n$ vertices has $O(nh\sqrt{\log h})$ edges.
\end{lemma}
We will use this result to force the presence of a $K_h$ minor even after deleting many vertices.  The following lemma allows us to densify a graph in terms of its girth (allowing us to use \autoref{lem: Kost1984} to argue the existence of a minor).  Recall that the girth of a graph is the length of the shortest cycle in the graph.  We give a tighter bound on the number of edges in a $K_h$-minor free graph with girth $g$ in \autoref{cor: at most n+O(n/g)}. The proof of \autoref{thm: lower_bound} is then concluded by finding a family of graphs that have sufficiently large girth.

\begin{lemma}\label{lem: |G'|=m-n+n'}
Let $G$ be a graph with $n$ vertices, $m$ edges, and sufficiently large girth $g$. Then it has a minor $G'$ that is a simple graph with $n'\le \frac{5n}{g}$ vertices and $m-n+n'$ edges.
\end{lemma}
\begin{proof}
Let $T$ be an arbitrary rooted spanning tree of $G$, let $r$ be the root of $T$, and let $V_i$ be the set of vertices at $i^{th}$ level of $T$. Let $\ell = \lfloor \frac{g-3}{4}\rfloor$.  We choose a non-negative integer $a < \ell$ such that
\begin{equation}
\mathcal{S} =r\cup\left\{\displaystyle{\bigcup_{k\ge 0}} V_{a+k \ell}\right\}
\end{equation}
contains at most $\frac{n}{\ell} \le \frac{5n}{g}$ vertices.  Set $\mathcal{S}$ is a collection of vertices at every
$\ell^{th}$ level starting from level $a$ along with root $r$.

Now we perform the following operation to obtain a minor $G'$ of $G$:
for every vertex $v \in G\setminus\mathcal{S}$ contract the edge $uv$
where $u$ is the parent of $v$ in $T$.  That is, for every $v\in V_i$,
where $i\neq a+k\ell$, we contract $v$ to its ancestor in $V_{i-1}$.
Since the distance between two consecutive levels of vertices in
$\mathcal{S}$ is $\ell$ (or less, as between $r$ and $V_a$) and the girth of $G$ is $g$, contracting these
edges cannot result in self-loops or parallel edges.  Therefore $G'$
is simple.

Since we contract $n-|\mathcal{S}| = n-n'$ edges, the
number of edges in $G'$ is $m- (n-n') = m-n+n'$.
\end{proof}

Consider a graph $G$ with $n$ vertices, maximum degree $3$, and girth $g$.  If $G$ has $n+\omega(\frac{n}{g}h \sqrt{\log h})$ edges, then, by \autoref{lem: |G'|=m-n+n'}, $G$ has a minor $G'$ with $O(\frac{n}{g})$ vertices and $\omega(\frac{n}{g}h \sqrt{\log h})$ edges.  By \autoref{lem: Kost1984}, $G'$ is dense enough to have a $K_h$ minor.  Therefore, we get:

\begin{corollary}\label{cor: at most n+O(n/g)}
Every simple $K_h$-minor-free graph $G$ with $n$ vertices, maximum degree $3$, and girth $g$ has $n+O(\frac{n}{g}h \sqrt{\log h})$ edges.
\end{corollary}

\paragraph{Proof of \autoref{thm: lower_bound}.}  Let $G=(V,E)$ be a 3-regular graph with $n$ vertices, $m = \frac{3n}{2}$ edges and girth $\Omega(\log n)$; for example, the Ramanujan graphs have this property~\cite{LPS88}.
In the following, we take $h$ to be a constant.  By \autoref{cor: at most n+O(n/g)}, $G$ has a $K_h$ minor.
Any subgraph $G^*$ (with $m^*$ edges and $n^*$ vertices) of $G$ also has girth $\Omega(\log n)$.  By deleting $k$ vertices, the best we can hope for is that we delete $3k$ edges.  That is, $m^* \ge m-3k$.  To ensure that $G^*$ does not have a $K_h$ minor, we need 
\[ \frac{3n}{2} - 3k = m-3k \le m^* \le n^* + O(n^*/g) = n-k+O\left(\frac{n-k}{\log n}\right)\]
Solving for $k$, we require that
\[k \ge \left(\frac{1}{4}-O(1/\log n)\right) n.\]
Substituting $2m/3$ for $n$ gives the theorem.

\bigskip

\paragraph{Acknowledgments} This material is based upon work supported by the National Science Foundation under Grant Nos.\ CCF-1252833 and CCF-1228639 and by the Office of Naval Research under Grant No. N00014-08-1-1015.  The authors thank Amir Nayyeri for helpful discussions.

{\raggedright
\bibliographystyle{abuser}
\bibliography{planarize}}

\end{document}